\documentclass{article}

\usepackage{graphics,epsfig,amsfonts,amsmath}
\usepackage{color}
\usepackage{amsmath}
\usepackage{amsthm}
\usepackage{amsfonts}
\usepackage{amssymb}
\usepackage{graphicx}
\usepackage{bbm}
\usepackage{float}
\usepackage{placeins}
\newtheorem{proposition}{Proposition}
\newtheorem{theorem}{Theorem}

\newtheorem{lemma}{Lemma}

\providecommand{\U}[1]{\protect\rule{.1in}{.1in}}
\newsavebox{\selvestebox}

\def\erre
{\hbox{\bf R}}

\def\a{\alpha}
\def\b{\beta}
\def\d{\delta}
\def\g{\gamma}

\def\l{\lambda}

\def\s{\sigma}

\def\cF{{\cal F}}

\def\cI{{\cal I}}

\def\cK{{\cal K}}
\def\cL{{\cal L}}

\def\cN{{\cal N}}
\def\cR{{\cal R}}

\def\cU{{\cal U}}

\def\acc{\`}

\def\dis{\displaystyle}

\def\dpar#1#2{\displaystyle{\partial{#2}\over {\partial {#1}}}}

\def\be{\begin{equation}}
\def\ee{\end{equation}}
\def\beq{\begin{array}}
\def\eeq{\end{array}}

\begin{document}

\title{Optimal epidemic control by social
distancing and vaccination of an infection structured by time since infection: the covid-19 case
study}
\author{Alberto d'Onofrio$^{1}$, Mimmo Iannelli$^{2}$, Piero Manfredi$^{3}$, Gabriela Marinoschi$^{4}$}
\date{}
\maketitle
\centerline{${}^1$ Dipartimento di Matematica e Geoscienze, Universit\acc a di Trieste} 
\centerline{Via Alfonso Valerio 12 Edificio H2bis, 34127 Trieste, Italy. e-mail: alberto.d'onofrio@units.it}
\centerline{${}^2$ Department of Mathematics, University of Trento} 
\centerline{Via Sommarive 14, 38123 Trento, Italy. e-mail: mimmo.iannelli@unitn.it}
\centerline{${}^3$ Dipartimento di Economia e Management, University of Pisa, Italy} 
\centerline{Via Ridolfi 10, 56124 Pisa--Italy. e-mail: manfredi@ec.unipi.it}
\centerline{${}^4$ Gheorghe Mihoc-Caius Iacob Institute of Mathematical Statistics and Applied Mathematics,}
\centerline{Romanian Academy, Bucharest, Romania. e-mail: gabriela.marinoschi@acad.ro} 


\maketitle

\begin{abstract}
Motivated by the issue of COVID-19 mitigation, in this work we tackle the general problem of optimally controlling an epidemic outbreak of a communicable disease structured by time since exposure, by the aid of two types of control instruments namely, social distancing and vaccination by a vaccine at least partly effective in protecting from infection. Effective vaccines are assumed to be made available only in a subsequent period of the epidemic so that - in the first period - epidemic control only relies on social distancing, as it happened for the COVID-19 pandemic. By our analyses, we could prove the existence of (at least) one optimal control pair, we derived first-order necessary conditions for optimality, and proved some useful properties of such optimal solutions. A worked example provides a number of further insights on the relationships between key control and epidemic parameters.
\end{abstract}


\section{Introduction}
The COVID-19 pandemic with its devastating effects on human societies, and the ensuing need for finding effective control tools during pandemic emergencies, has shown how dramatically poor still is humans' ability to control the pervasive societal impact of serious pandemic events. 
Among other, the COVID-19 pandemic has opened an endless list of new questions and problems in mathematical epidemiology, in particular in the area of optimal control theory applied to communicable diseases. With hindsight this area, which covers a wide disciplinary range, from theoretical epidemiology to engineering, mathematics and economics, was constrained to somewhat narrow boundaries before the COVID-19 pandemic. Indeed, though in the pre-COVID era many excellent works had been produced of which we can cite here here only a sample \cite{wickwire1975optimal, greenhalgh1988some, behncke2000optimal, gaff2009optimal, jung2009optimal, rowthorn2009optimal, hansen2011optimal, klepac2011synthesizing,klepac2012optimizing,klepac2016self, ainseba2012optimal, rowthorn2012optimal, betta2016perspectives, betta2019models, buonomo2019optimal, bussell2019applying}, the field suffered the implications of the prevailing \emph{reductionist} approach, mainly focusing on the containment of the \emph{direct impact and costs} of the infection on health and therefore disregarding any secondary, or indirect effects on health, the economy, and the society as a whole.\\
The COVID-19 experience has brought a revolution into this. Indeed, the nasty perspectives suddenly opened in early spring 2020 in countries that first faced the COVID-19 threat (as China and Italy) by the saturation of hospitals on the one hand and by the dramatic societal impact of persistent generalised lockdowns, have highlighted the dramatic trade-off between saving human lives from the disease vs protecting the society as a whole (and therefore human lives) from the eventual impact of the restrictions needed to contain infection spread. The list of indirect effects of the pandemic through the control measures enacted to mitigate it, is an endless one and ranges from economic effects (e.g., repercussions on employment, supply effects on production, demand effects, deflation, income decline etc), to secondary health effects (e.g., postponing treaments and prevention with consequences for chronic diseases and other non-communicable conditions, such as the consequences of isolation on mental health), to socio-educative-economic (e.g., think to the complex chain of implications due to school closure impacting on individual's education, on the overall generation of human capital, on inequality etc), to relational ones, etc.\\
The first papers attempting to account for the broader view of optimal control applied to epidemic outbreaks were from economists and we report here only the earliest of a long list \cite{alvarez2020simple,acemoglu2021optimal,eichenbaum2021macroeconomics,glover2020health}. Obviously, many other approaches were also proposed for this problem (see e.g., \cite{elie2020contact} and references therein).\\
This work aims to contribute to this broader perspective by analysing, from a theoretical standpoint, the general problem of optimally controlling an epidemic of a communicable disease structured by \emph{time since infection} through two main control instruments that is, social distancing and vaccination. However, following the realm of COVID-19 (and of any future pandemic event), social distancing (and in general non-pharmaceutical interventions) is assumed to be the only available control tool in a first period of the epidemic due to delayed arrival of the vaccine. The availability of vaccines in a subsequent period, expands the range of control tools and therefore of possible interventions, potentially allowing the gradual relaxation of social distancing and the restart of full socio-economic activity without compromising the possibility to bring the epidemic under full control, up to possible (local) elimination. 
As for the epidemic model, we consider an infectious disease structured by time since infection and characterised by waning immunity of both recovered and vaccinated individuals. Waning of immunity, both natural and vaccine related, has been a critical feature of COVID-19 even in the absence of variants \cite{milne2021does} and makes it unavoidable to adopt a susceptible-infective-removed-susceptible (SIRS) framework. Structuring the model by time since infection via a generic infectiousness kernel, confers a great deal of flexibility to the model form, for example to reflect the presence of a latent, pre-infective, phase (see \cite{d2021dynamics} and references therein) as well as other complications.

Further, the model rests on a simplifying hypothesis that we termed the 'low attack rate' hypothesis, LAR. The LAR hypothesis captures a specific set of real-world settings or specific situations/phases of the COVID-19 pandemic for which robust empirical evidence is available that the impact of the epidemic has been effectively mitigated by social distancing (and other non-pharmaceutical interventions, NPIs) in the period preceeding the arrival of vaccines, and by a combination of NPIs and vaccination after vaccine arrival. For example, in Italy a nation-wide large-scale serological survey conducted in June-July 2020, after the first wave and the generalised lockdown adopted to mitigate it were over (since May 5th 2020) showed that the national average attack rate was only slightly in excess of $2\%$. This figure was representing an average of a widely heterogeneous situation, with some areas (mainly within Lombardy region, Northern Italy) which suffered strong attack rates (with a regional average of $7.5\%$) but most other regions of the country were only very mildly attacked \cite{istat2020}. 
Outside Italy,  the LAR hypothesis closely mirrors China 'total COVID-19 control' policy. However, it also mirrors those countries (e.g., Iceland, Denmark, New Zealand, Southern Korea, Japan and Australia) that, during the first pandemic year, opted for a policy of strong control (termed,  elimination) rather than just mitigation (see \cite{oliu2021sars}). Notably, in these countries, the achievement of the best COVID-19 control performances went parallel with the lower societal damage \cite{aghion2021aiming,oliu2021sars}.

The way the LAR hypothesis will be modelled is by assuming that the susceptible fraction of the population remains close to $100\%$ prior to vaccine arrival, and that- after vaccine arrival - the susceptible fraction is mostly depleted by vaccination and only negligibly by infection. 

Clearly, during the COVID-19 pandemic, the LAR hypothesis became less and less valid as soon as variants of concern as \emph{delta or omicron}, capable to elude vaccine protection against infection, became dominant. However, we feel that the LAR hypothesis is much more than a useful simplifying hypothesis. Indeed, containment and mitigation of infection incidence at manageably low levels - what the LAR hypothesis describes in mathematical terms - is the primary form of intervention that governments and public health institutions can use to mitigate the direct and indirect (i.e., societal) impacts of serious (current and future) pandemic threats. Additionally, in our opinion, the tools to keep most pandemics under LAR boundaries are available, namely by early interventions and their boosting. Therefore, the LAR hypothesis should also represent - by itself - a fundamental policy prescription for policy makers thinking to future pandemic preparedness plans. 

The proposed model shows a \emph{multi-period} structure due to the presence of  (i) an initial period of uncontrolled epidemic growth, (ii) a subsequent period where social distancing (and in general non-pharmaceutical measures) are the only available control tool, (iii) a period where effective vaccines becomes available. In particular, the considered optimal control problem is finite-horizon where the \emph{switch times} between the various periods are taken to be (at least approximately) known to the policy maker.
After some basic characterization, the model is set into an optimal control framework, with the main goal to investigate the optimal combined pathway of social distancing and vaccination in a multi-epoch framework. The cost functional includes direct epidemiological costs, related to infection prevalence, as well as indirect costs of the epidemic. 
As main results of our analysis, we proved the existence of a joint optimal control and provided a characterization of the necessary conditions on the controls to be used in realistic applications for the effective determination of an optimal strategy.
The manuscript is organised as follows. Section \ref{epi_model} introduces the epidemiological model in its full generality and presents its multi-epoch form. Section \ref{ocp} presents the full optimal control problem, while section \ref{existenceocp} analyses the existence of the joint optimal control. Section \ref{opcond} characterises the necessary conditions for optimality. A number of remarks are discussed in \ref{remarks}. Discussion and conclusions follow.

\section{The general epidemiological model}\label{epi_model}
As previously pinpointed, we consider a model of a serious but non-fatal communicable disease imparting temporary immunity (both natural and vaccine-related) and structured by time since infection. The population is assumed to be constant over time and having size $N$ under outbreak conditions i.e., we disregard vital dynamics. As it has been the case for COVID-19, the infection is assumed to be controllable by two main types of interventions namely 
\begin{itemize}
\item[(i)]non-pharmaceutical interventions, as represented by social distancing, mainly
affecting the transmission rate $\beta(t)$ over time,
\item[(ii)] vaccination, according to an appropriate time-varying $V(t)$ function representing the absolute rate of effectively immunised individuals at time $t$.
\end{itemize}
Note that the modulation of interventions by time-varying functions allows to include realistic COVID-19 features such as e.g., the fact that social distancing was the only available intervention in a first phase of the epidemic, while both social distancing and vaccination coexisted in a later stage.  For sake of simplicity we disregard further interventions such as testing/tracing and isolation of infective individuals as the latter can be 
easily incorporated into the removal process.

The proposed model has an SIRS structure extending the SIR model in \cite{d2021dynamics} and it is based on a state variable, structured by \emph{time since infection} $x$
\be\label{density}
Y(x,t),\ \  x\in [0,x_+]\ , t\ge 0 ,
\ee 
representing the (improper) $x$-density (i.e., not normalised) of infected individuals at time $t$, and on the two homogeneous (i.e., unstructured)  variables
\[
S(t), \quad R(t)\qquad t\ge 0
\]
respectively representing the number of susceptible  and the number of removed individuals. 
 The internal variable $x$ of function $Y(x,t)$ namely the time since infection, represents the time elapsed since the individual has been infected i.e., the duration of sojourn of infective individuals in the infected state, conferring to the model the so-called semi-Markov form. After the seminal work by Kermack and McKendrick \cite{kermack1927contribution}, from which the present model directly stems, the time since infection variable has been customarily termed \emph{age of infection} in epidemiological modeling (e.g., \cite{diekmann1990definition,diekmann2000mathematical,thieme1993may}). Structuring the infected population by the age of infection $x$, allows a fine description of those internal processes, namely the onset and development of infectivity and its subsequent decline until removal with acquisition of immunity, which critically shape the generation of secondary cases. In particular, in \eqref{density} $x_+$ denotes the maximal time since infection and may be finite or infinite; in our general analysis we will focus on the more realistic case $x^+ < \infty$ though we will also present a worked example allowing the maximal age to be infinite in order to allow the reduction of our general problem to ordinary differential equations. 

The model equations read as
\be\label{model}
\left\{
\beq{rl}
i)&\dis S'(t) = - \frac{S(t)}{N} \int_0^{x_+} \b(t,x) Y(x,t) dx - V(t) + \d R(t)\\\\
ii)&\dis\left\{
\beq{l}
\dis \left(\dpar t{} + \dpar x{}\right) Y(x,t) = -\g(x) Y(x,t) ,\\\\
\dis Y(0,t)= \frac{S(t)}{N}  \int_0^{x_+} \b(t,x) Y(x,t) dx ,
\eeq
\right. \\\\
iii)&\dis R'(t) = \int_0^{x_+} \g(x) Y(x,t) dx  + V(t) - \d R(t)
\eeq
\right.
\ee
endowed with the initial conditions
\be\label{initial}
S(0)=S_0,\quad Y(0,x)= Y_0(x),\quad R(0)=R_0 .
\ee

The parameters in \eqref{model} have the following meaning

\begin{itemize}
\item
$\b(t , x)$ =  the transmission rate of infected individuals with time since infection $x$ at time $t$\\ which represents the expected number of secondary infections caused by a single infected with time since infection $x$ per unit of time in a fully susceptible population i.e., for $S(t)/N \approx 1$; clearly, if the susceptible fraction is depleted by the epidemics the number of secondary cases of an infective with time since infection $x$ will have to be discounted by the susceptible fraction.
\item
$\gamma(x)$  = the removal rate of infected individuals with time since infection $x$ at time $t$\\  i.e., the overall per-capita rate at which infected individuals with time since infection $x$ are removed by any cause as e.g., recovery or death (during an uncontrolled epidemic) or by screening, tracing, isolation, hospitalization etc, in the presence of interventions; for simplicity this removal rate is assumed to be independent of time as we suppose that the conditions that contribute to removal during the epidemic phases considered in this work remain essentially unhaltered over time. This does not need being true during the early epidemic phase but this is not relevant for our purposes here; note that $\g(x)$ is related to the probability to be still infected after a time $x$ since infection (the so-called \emph{survival-to-removal} function) by the equation
\[
\Gamma (x) = e^{-\int_0^x \g(s) ds};
\]
\item
$V(t)$ = the absolute instantaneous rate of successful immunization per unit of time\\
i.e., coarsely speaking, the number of vaccinated susceptibles that are effectively protected against infection (and not just from the disease) including full vaccine failure. Function $V(t)$ is assumed to be identically zero prior to the arrival of the vaccine. According to the experience of some countries, which were the first that started mass vaccination campaigns against COVID-19 in view of their market power, after the arrival of the vaccine this function has shown a range of increasing shapes, reflecting logistic difficulties in an early phase, and then plateauing due to achievement of maximal capacity.
The only requirement on $V(t)$ is that it does not violate the positivity of $S(t)$.
\item
$\delta$ = the immunity waning rate \\ i.e., the rate at which removed people loose their immunity and become again susceptible; for the sake of simplicity $\delta$ is taken here to describe both natural and vaccine waning immunity.
\end{itemize}
Concerning $\b(t,x)$ we assume the following constitutive form
\[
\b(t,x)= c(t)\b_0(x),
\]
where the factor $c(t)$ denotes the number of adequate contacts per person and per unit time, while $\b_0(x)$ tunes the (average) intrinsic infectiousness of an infected individual with time since infection $x$ (i.e, the probability that an infected
individual with time since infection $x$ infects a susceptible during an adequate contact), possibly related to her/his viral load.  We deliberately did not include time dependencies in the infectiousness due to individual protections measures as e.g., wearing masks, given that we can simply assume that the latter measures scale the adequate number of contacts and are therefore embedded into $c(t)$.

A useful remark on previous hypotheses is that considering a structured model with time since infection and generic dependent infectivity/recovery kernels, confers a great deal of flexibility to the model itself. For example, this allows to reflect the possible presence of an exposed ('E'), pre-infective, period (see \cite{donofrioiannellimanfredi2022} and references therein) as well as other modeling complications. Therefore the intimate structure of the proposed model can be considered of the SEIRS type.

The effective reproduction number of the infection at a generic time $t$ reads
\[
\cR(t) = c(t) \frac{S(t)}{N}\int_0^{x_+} \b_0(x) \Gamma(x) dx .
\]

Model \eqref{model}  is non-linear because of the first term in $i)$ and the boundary condition included in $ii)$ where the incidence of new infections per unit of time
\be\label{incidence}
U(t)= \frac{S(t)}{N} \int_0^{x_+}\!\!\!\b(t,x) Y(x,t) dx = c(t) \frac{S(t)}{N} \int_0^{x_+} \b_0(x) Y(x,t) dx
\ee
depends both on $S(t)$ and $Y(x,t)$. However, as stated in the Introduction, we will adopt an approximated version of the previous model, based on the concept of an adequately contained outbreak, what we described as the LAR hypothesis \cite{d2021dynamics}. This amounts to assume that: 
\begin{itemize}
\item[(i)] prior to the introduction of mass vaccination an early and effective social distancing activity had been undertaken bringing the current reproduction number of the infection $\cR(t)$ below one or at most in the region of one, so that the overall epidemic attack is unable to bring a serious 
depletion of the susceptible population, 
\item[(ii)] after the introduction of vaccination, the epidemic remains adequately controlled through a combination of (possibly relaxing) social distancing and immunization; obviously, in this phase vaccination becomes the key responsible for the depletion of the susceptible population, while the contribution from infection remains moderate. 
\end{itemize}
On these assumptions, the vaccination program starts when the susceptible population is still essentially nearby 100\% while, during the vaccination program progression, the removed population is essentially composed of vaccinated individuals. For COVID-19, this has surely been the case - besides China - for those countries worldwide as e.g., Iceland, Denmark, New Zealand, Southern Korea, Japan and Australia that were able to enact an early and highly effective social distancing policy \cite{oliu2021sars}.  

With these premises,  we disregard the first terms in equations (\ref{model},$i$), (\ref{model},$ii$), so that, summing the two equations, we get   $S(t) + R(t)=N$. Then, 
by denoting as $\dis s(t)=\frac{S(t)}N$ the susceptible fraction and by $\dis v(t)=\frac{V(t)}N$ the number of daily successful immunizations per head, the basic model becomes
\be\label{modelreduced}
\left\{
\beq{rl}
i)&\dis s'(t) = -v(t) + \d (1 -s(t))\\\\
ii)&\dis \left(\dpar t{} + \dpar x{}\right) Y(x,t) = -\g(x) Y(x,t),\\\\
iii)&\dis Y(0,t)=  c(t)s(t) \int_0^{x_+} \b_0(x) Y(x,t) dx ,
\eeq
\right.
\ee
where $c(t)$ and $v(t)$ play a role of control variables and are designed according to the available control strategies. The system is endowed with the pandemic initial conditions 
\be\label{pancond}
s(0)=1,\quad Y(x,0)=Y_0(x).
\ee

Note that the susceptible equation in the previous formulation is linear and has a full closed form analytical solution depending on the actual time-dependency of the vaccination rate function. Indeed, using the initial condition in \eqref{pancond}, we have
\begin{equation}\label{s}
s(t)=1 - \int_{0}^{t}e^{-\delta (t-\s)}v(\s) d\s,\quad t\in [0 , T] .
\end{equation}
Though $s(t)$ is not necessarily decreasing we have $s(t) \le 1$. Moreover, in order to guarantee $s(t)$ to be strictly positive we adopt the following constraint on $v(t)$
\be\label{deltav}
v(t) \le \delta,\quad  t\in [0 , T].
\ee
Indeed, with this constraint, from \eqref{s} we have
\[
s(t) \ge e^{-\delta t} > 0 .
\]

Problem \eqref{modelreduced} represents the state systyem for our epidemic optimal control problem that will be formulated in the following sections.  Concerning the basic parameters appearing in \eqref{modelreduced}, we will consider the following mathematical assumptions
\be\label{beta}
\b_0 (x )\ge 0, \quad  \b_0 \in  W^{1,\infty}(0,x^+) ,
\ee
\be\label{gamma}
\g (x) \ge 0, \quad  \g \in L^1_{loc}(0,x^+) ,\quad \int_0^{x^+}\!\!\! \g (\s) d\s = +\infty,\quad \g(\cdot)\Gamma(\cdot) \in L^\infty(0,x^+) ,
\ee
\be\label{Y0}
Y_0(x) \ge 0, \quad  \frac{Y_0(\cdot)}{\Gamma(\cdot)} \in  L^\infty(0,x^+) .
\ee
Note that condition  \eqref{gamma} implies $\dis\Gamma(x^+)=0$, as expected.

\subsection{Multi-period epidemics}\label{Multiepoch}
Given the multi-period setting of our problem motivated by COVID-19, we consider an epidemic outbreak along a fixed time range $[ -T_0 , T]$, divided into three periods as follows:
\begin{itemize}
\item[0)] a first period (actually named Period 0) of free (i.e., uncontrolled) epidemics early growth, taking place over a time interval $[-T_0,0]$, during which no control measures are enacted and the epidemics follows its natural progression; this phase is characterized by 
\[
c(t)=c_0 ,\quad v(t)=0, \quad s(t)= 1,
\]
and we have
\[
\cR(t)= c_0 \int_0^{x_+} \b_0(x) \Gamma(x) dx = \cR^0_0 ,
\]
i.e., the contact rate is unaltered at its normal level in the absence of any epidemic alert. In this period the current reproduction number $\cR(t)$  coincides with the the basic one, summarizing epidemic growth per generations in a wholly susceptible population in the absence of intervention measures. Further, following \cite{d2021dynamics}, we suppose that the initial infective cohort is distributed according to the stable distribution corresponding to the given epidemiological parameters  
\be\label{stabledistribution}
Y^0_*(x)=\frac{e^{-\a_0^*x}\Gamma(x)}{\dis \int_0^{x_+} e^{-\a_0^*x}\Gamma(x) dx}, \quad x\in [0,x_+]
\ee
where $\a_0^*$ is the leading root of the characteristic equation
\[
c_0\int_0^{x_+} \b_0(x) \Gamma(x) e^{-\l x} dx  =1 .
\]

The previous hypotheses imply that $Y_0(x)= I_0 \ Y^0_*(x)$, where $I_0$ is the initial  total number of infective individuals, and yields, for $t\in [-T_0 , 0]$,
\be\label{primafase}
\beq{c}
\dis Y(x,t) = e^{\a_0^*( t+T_0)}\  Y_0(x),\\\\
\dis U(t)= \frac{I_0\ e^{\a_0^*( t+T_0)}} { \int_0^{x_+} e^{-\a_0^*x}\Gamma(x) dx},\quad I(t)= I_0\ e^{\a_0^*( t+T_0)} .
\eeq
\ee

\item[I)] a second Period (actually Period 1) relative to the time interval $[0 , T_1]$, $T_1 <T$ when only social distancing measure are enacted to mitigate the epidemics; during this phase $c(t)$ is tuned according to the established social distancing strategy and again
 \[
v(t)=0,\ \hbox{and}\  s(t)=1;
\]

\item[II)] a third Period (Period 2) displaying through the interval $[T_1 , T]$, during which also the vaccination campaign is enacted and coexists with social distancing measures, potentially allowing for the possibility to gradually relax social distancing measures; in this case $c(t)$ and $v(t)$ depend on time according to the adopted strategy and $s(t)$ is consequently calculated through equation (\ref{modelreduced}, $i$).
\end{itemize}

 \textbf{Remark: switch times}. In this effort, the switch times between interventions, or equivalently the duration $T_0, T_1, T-T_1$ of the different  periods are taken as known. This hypothesis is clearly a departure point. The duration of the first period essentialy depends on the  time necessary to detect the epidemic in a form out-of-control and the subsequent policy delay to enact the corresponding control measures \cite{carrozzo2020deteriorated,diekmann1990definition,thieme1993may,diekmann2000mathematical}. The switch between Period 1 and Period 2 depends on the sum of the vaccine development time plus the time necessary to make available a sufficiently large stock of vaccine to adequately start the campaign. The possibility to treat duration $T_1$ as random variable was considered in some main economic efforts \cite{alvarez2020simple,acemoglu2021optimal}.  

\section{The optimal control problem}\label{ocp}

Within the previous multi-phasic epidemic setting, our general aim is to determine  the optimal couple $(c, v)$ minimizing the following cost function
\be\label{cost}
\beq{l} 
\dis \Phi(c,v) = A_0 \int_0^{T} c(t) s(t)\int_0^{x_+} \b_0(x) Y(x,t) dx dt  \\ \\
\dis \hskip3cm +  A_1\int_0^{T}  Q\left(c(t)\right)  dt+ \int_0^{T} \left(A_{21} v(t)+ \frac{A_{22}}{2} v^2(t) \right) dt
\eeq
\ee
in the set of controls
\be\label{setcontrol}
\cU \equiv \left[ 
(c(\cdot), v(\cdot))\ \ \hbox{such that}\ \ \left\{
\beq{ll} 
\dis c_- \le c(t) \le c_0,\ t \in [0 , T] ;\\\\
\dis 0 \le v(t) \le \delta ,\ t\in [0 , T];\\ 
\dis v(t)=0\ \hbox{for}\   t \in [0 , T_1) 
\eeq
\right.
\right] .
\ee

The first term in \eqref{cost} accounts for the direct health cost of the epidemic arising from the overall epidemic incidence. This aims to proxy the true costly direct events of the epidemic such as e.g., hospitalizations and deaths, which have not been explicitly introduced in our model just for sake of simplicity. 

The second term represents the indirect epidemic cost, namely the broader societal and economical cost due to the social distancing activities and their impact on the labour force, production, on the overall social and relational activities and also - still indirectly - on health, due to the declined prevention and medical activities caused by the restrictions. The previous formulation is a highly parsimonious one reflecting the fact that our communities 'function' because of social activities i.e., thanks to the social contacts embedded in function $c(t)$, and therefore that any restriction to such contacts, as those due e.g., to social distancing, cause a cost to the community as a whole. More explicit parametrizations of this term have been proposed in the COVID-19 economic literature \cite{alvarez2020simple,acemoglu2021optimal}, but we preferred the present one as more compact and parsimonious.
The hypothesised cost is driven by the function $Q(x)$ satisfying\be\label{Q}
Q(x)\ \hbox{is a l.s.c convex function on}\ \ (0 , c_0],\quad \lim_{x \rightarrow  0^+} Q(x)=+\infty ,\quad Q(c_0)=0 .
\ee
Indeed, the second property in  \eqref{Q} means that 
reducing to zero social contacts would destroy the economy, bringing, thus, very large costs (indeed we will keep the number of contacts above the threshold $c_- > 0$ which is considered acceptable); instead, the second condition means that keeping contacts at the level $c_0$, i.e. the typical level in the absence of epidemic restrictions, does not produce costs for the economy. Parametrized classes for $Q(x)$ are, for instance ($x\in [c_- , c_0]$, $Q_0>0$)
\[
Q(x)= Q_0  (c_0 - x)^{\rho}  \ (\rho>1) , \quad Q(x)= Q_0 \frac{c_0^{\rho} - x^{\rho}}{x^{\rho} c_0^{\rho}} \ (\rho > 0),
\]
shown in Figure \ref{figure1} for different values of the parameters. The stated convexity assumption (enhanced by the parameter $\rho$) is technically convenient for later developments but, in addition, is especially relevant on the substantive side. Indeed, as is the case of previous simple parametric families of functions, it allows to define costs components that (i) are not too sensitive when the current contact rate $x$ is close to the 'normal' contact rate $c_0$ (corresponding to the situation of absence of epidemic alert), meaning that small deviations from $c_0$ might be sustainable from the broader societal and economic point of view; (ii) generate dramatically high societal costs when large deviations of the contact rate from its normal value are considered. This is consistent with the aforementioned idea that large social distancing interventions, breaking social and relational ties, will bring dramatic costs for the society as a whole. 

\begin{figure}[h]
\includegraphics[width=6.5cm]{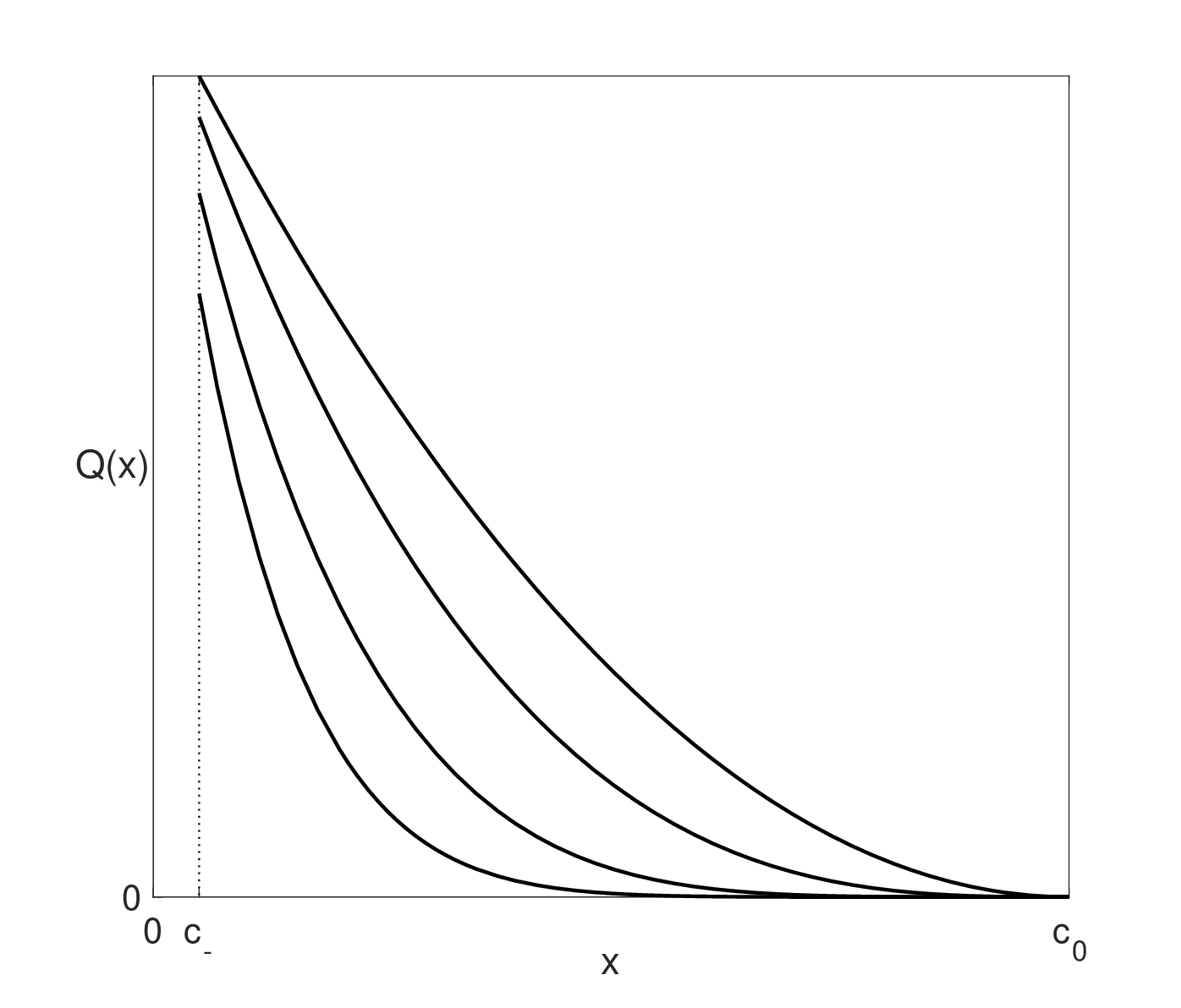}\includegraphics[width=6.5cm]{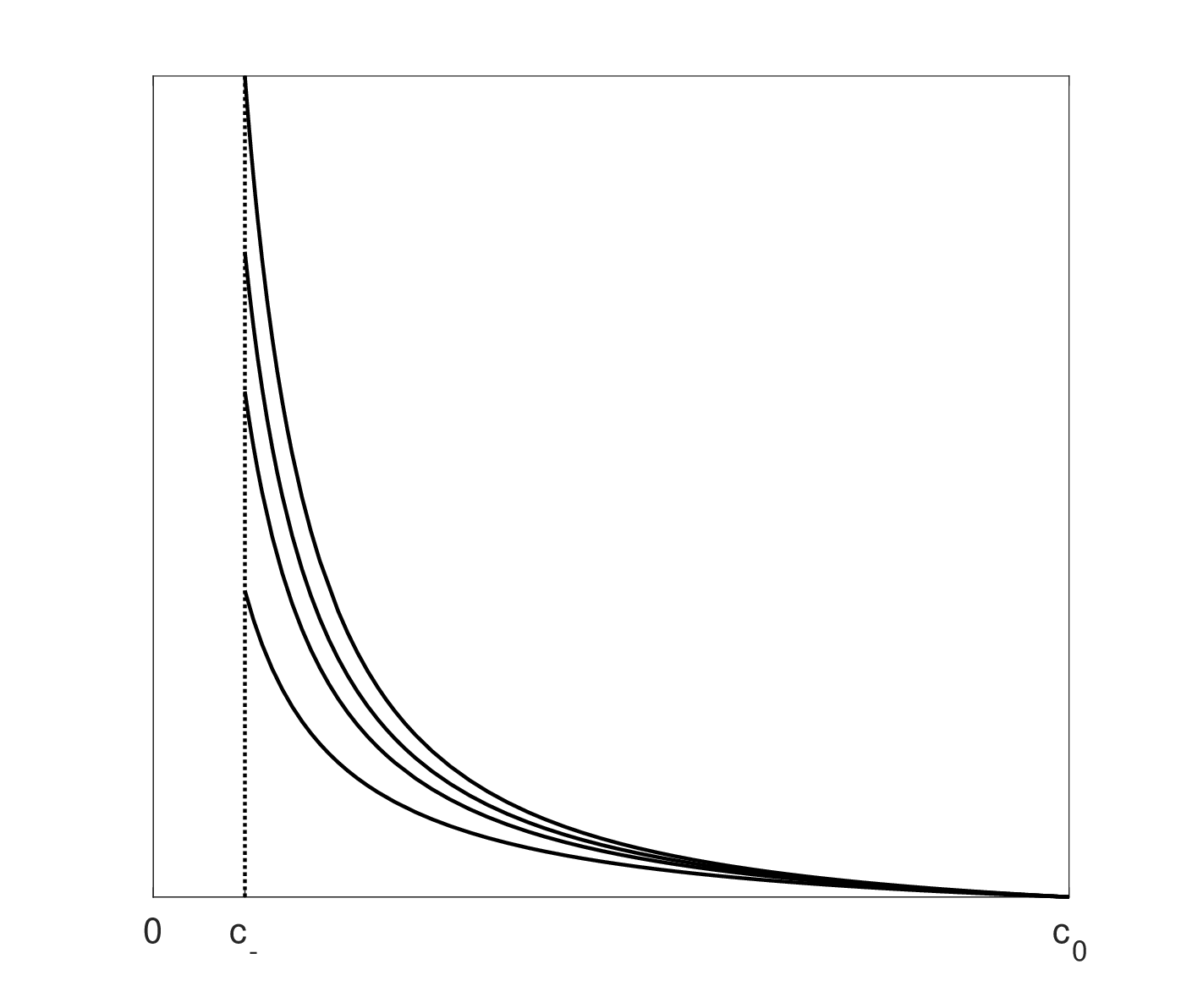}
\caption{\label{figure1}Parametrized families for function $Q(x)$ (restricted to the interval $[c_- c_0]$). Left panel: the function $Q(x)= Q_0  (c_0 - x)^\rho$ for different values of the parameter $\rho$, higher values of $\rho$ correspond to lower curves. Right panel: the function $Q(x)= Q_0 \frac{c_0^\rho - x^\rho}{x^\rho c_0^\rho}$ for different values of $\rho$, higher values of $\rho$ correspond to higher curves.}
\end{figure}

Finally, the third term in \eqref{cost} accounts for the cost of the efforts undertaken for the vaccination campaign.  The linear term, depending on coefficient $A_{21}$, corresponds to the idea that this cost is essentially proportional to the number of vaccinations administered during the campaign; the quadratic term (depending of coefficient $A_{22}$) takes into account the fact that achieving a high vaccination coverage produces additional costs for the system, primarily due to the difficulties related to reaching marginal subpopulations, for which there is a large evidence in the literature. The latter case does not simply include strongly hesitant people e.g. 'anti-vaccinators', but simply less informed groups such as e.g., foreign-born people, low-literacy groups etc \cite{ macdonald2015vaccine}. In some exceptional circumstances, a negative quadratic term might be considered e.g., (including also the vaccine price) under special forms of the vaccine demand function by countries (e.g., some countries could join in a 'cartel' to obtain a much lower price under the promise to buy larger amounts). However, we will not pursue this case in this work.  

Based on previous discussion, we assume that
\[
0 <  A_{22}  <<<  A_{21} .
\]

All the constants, $A_0$,  $A_1$, $A_{22}$, $A_{21}$, are coefficients that state the relative weights of the different cost items. 

The set $\cU$ reflects the partition into epidemic phases discussed above, and belongs to the space $L^\infty(0,T)\times L^\infty(0,T)$; actually we have that $\cU$ is weak-* closed  in $L^\infty(0,T)\times L^\infty(0,T)$ and that on $\cU$ the function $\Phi (c,v)$ is not identically infinite, in fact if we choose 
\[
c(t)=c_{0},\ t\in [0,T],\quad v(t)=0,\ t\in [0,T],
\]
we see that $\Phi (c,v)<\infty $. 

The previous problem represents a fairly general approach, to the best of our knowledge fully new, which considers the possibility to identify global optimal control strategies throughout the entire course of a multi-phasic epidemic including 
\begin{itemize}
\item[(i)] an optimal social distancing schedule in the absence of vaccination, 
\item[(ii)] an optimal (relaxation of) social distancing after the arrival of vaccination. 
\end{itemize}
This general problem includes a number of interesting subcases. These subcases are discussed later in Section \ref{remarks}.

\section{Existence of an optimal control}\label{existenceocp}
In this section we shall prove the existence of the solution to the state
system and the existence of a solution to the optimal control problem. Indeed, for any input $(c, v) \in \cU$,  equation $i)$ in \eqref{modelreduced} has the solution \eqref{s} independently of the other equations $ii)$ and $iii)$. Concerning these latter, they are a variation of the classic Lotka-McKendrick equation (see \cite{iannelli2017basic}) and have a solution given by
\begin{equation}
Y(x,t)=\left\{ 
\begin{tabular}{ll}
$c(t-x)s(t-x)\ Z(t-x)\ \Gamma (x),$ & $x\leq t$ \\ \\
$Y_{0}(x-t)\dis\frac{\Gamma (x)}{\Gamma (x-t)},$ & $x>t$
\end{tabular}
\right. ,  \label{Y}
\end{equation}
where
\be\label{zeta}
 Z(t)=\int_{0}^{x_{+}}\beta _{0}(x)Y(x,t)dx
\ee
solves the integral equation
\begin{equation} \label{B}
Z(t)= \int_{0}^{t} \mathcal{K}(t-x,x) Z(t-x)dx+F(t), \quad t\in [0,T]
\end{equation}
with
\[
\beq{ll}
\mathcal{K}(\s,x )=c(\s) s(\s) \beta _{0}(x) \Gamma (x)\\\\
\dis F(t)= \int_{0}^{\infty }\beta _{0}(x+t)  \Gamma (x+t)\frac{Y_{0}(x)}{{\Gamma (x)}}dx ,
\eeq
\]
where $ \beta _{0}(x)$, $ \Gamma (x)$ and $Y_0(x)$ are extended by zero outside the interval $[0,x^+]$.

Note that, under our assumptions \eqref{beta}-\eqref{Y0}, $\b_0$ and $\Gamma$ are Lipschitz continuous, thus
\[
\left| \b_0(x_1)\Gamma (x_1) - \b_0(x_2)\Gamma (x_2)\right| \le L | x_1 - x_2| ,
\]
and,  for any $(c, v) \in \cU$,  we have
\be\label{KF}
\beq{ll}
\mathcal{K}(\s,x ) \ge 0,\quad \left |\mathcal{K}(\s,x_1 )-\mathcal{K}(\s,x_2 )\right| \le  c_0\  L \ | x_1-x_2| \\\\
F(t) \ge 0,\quad  \dis \left | F(t_1) - F(t_2) \right| \le  L \left\| \frac{Y_0(\cdot)}{\Gamma(\cdot)} \right\|_1 \ |t_1 - t_2| .
\eeq
\ee

The following theorem recalls simple results from the theory of  integral equations of Volterra type, to be used in connection with \eqref{B} and with later developments
\begin{theorem}\label{propVIE}
Let 
\[
k(\s,x) \in L^\infty((0,T)\times (0,+\infty) ),\quad f(t)\in C([0,T]),
\]
then the integral equation
\be\label{VIE}
\dis u(t) = \int_0^t k(\s , t- \s ) u( \s) d\s + f(t)
\ee
has a unique solution $u \in  C([0,T])$. If moreover $k^\l$ and $f^\l$ are sequences such that
 \[
 \lim_{\l \to 0}  k^\l = k \  \ \hbox{in}\  \ \ L^\infty((0,T)\times (0,+\infty) ), \quad  \lim_{\l \to 0}  f^\l  = f \  \ \hbox{in} \ \ \  C([0,T]),
 \]
and $u^\l$ is the solution relative to $k^\l$ and $f^\l$, then
\be\label{limVIE}
 \lim_{\l \to 0} u^\l = u\   \ \hbox{in} \ \ \  C([0,T]).
 \ee
 \end{theorem}
\begin{proof}
To prove existence we may use a standard argument based on the iterates
\[
\beq{l}
\dis u_0(t) = f(t) , \\\\
\dis u_n (t) =   \int_{0}^{t} k(x,t-x) u_{n-1}(x)dx+f(t) \quad (n \ge  1).
\eeq
\]
Indeed, using standard estimates (see for instance \cite{iannelli2017basic}) we may prove that $u_n \in C([0,T])$  and that  the sequence $u_n (t)$ converges in $C([0,T])$ to a  function $u(t)$ which is the unique solution to equation \eqref{VIE}. Besides, by Gronwall lemma we have the following estimate
\[
|u(t)| \le \left\| f\right\|_\infty e^{\left\| k\right\|_\infty T},
\]
so that, using again Gronwall lemma, we have
\[
|u(t) - u^\l(t)| \le  e^{\left\| k\right\|_\infty T} \left\| f - f^\l \right\|_\infty + T  e^{\left\| k^\l\right\|_\infty T} \left\| k - k^\l \right\|_\infty
\]
and \eqref{limVIE} follows.
\end{proof}

We use Theorem \ref{propVIE} to state precisely existence and properties of the solution of  \eqref{B}  under our assumptions: 
\begin{proposition}
Let assumptions \eqref{beta}-\eqref{Y0} hold, then equation \eqref{B} has a unique solution $Z \in W^{1,\infty}(0,T)$ such that
\be\label{stime}
0 \le Z(t)\le M ,
\ee
\be\label{stime1}
\left | Z(t_1) - Z(t_2) \right| \le M_L |t_1 - t_2| ,
\ee
where
\[
\beq{l}
\dis M =  \|\b_0\|_\infty \left\| Y_0\right\| e^{c_0 \| \b_0\|_\infty  T},\\\\
\dis M_L =\left( L  \left\| \frac{Y_0(\cdot)}{\Gamma(\cdot)} \right\|_1 + (L T +  \|\b_0\|_\infty)  \|\b_0\|_\infty \left\| Y_0\right\| e^{c_0\| \b_0\|_\infty  T} \right) .
\eeq
\]
\end{proposition}
\begin{proof}
Noticing that $F(t)$ is continuous, existence and uniqueness hold by Theorem \ref{propVIE}. In addition, by \eqref{KF} we have that the iterates mentioned in the proof are non negative so that also $Z(t)$ is. Moreover, from \eqref{B} we have
\[
0 \le Z(t) \le  \left \| F(\cdot)\right \|_{\infty} + c_0 \|\b_0\|_\infty  \int_0^t Z(s) ds  \le  \|\b_0\|_\infty \left\| Y_0\right\| + c_0 \|\b_0\|_\infty  \int_0^t Z(s) ds 
\]
and \eqref{stime} follows by Gronwall lemma. Finally, since
\[
\beq{l}
\dis \left | Z(t_1) - Z(t_2) \right | \le \left| F(t_1)- F(t_2 )\right| \\\\
\dis \hskip1cm + \int_{0}^{t_1} \left |\mathcal{K}(x , t_1-x) - \mathcal{K}(x , t_2-x)\right |  Z(x) dx + \left| \int_{t_1}^{t_2} \mathcal{K}(x , t_2-x) Z(x) dx\right| ,
\eeq
\]
using \eqref{KF}, we have \eqref{stime1}.
\end{proof}

Once we have the solution $Z(t)$ to problem \eqref{B} we get the solution to \eqref{modelreduced} via formula \eqref{Y}. Indeed we can directly check  that the function $Y(x,t)$ given by \eqref{Y} is the solution of (\ref{modelreduced}, $i$), satisfying condition (\ref{modelreduced}, $ii$). In addition we have
\be\label{Ystima}
\left| Y(x,t) \right|\  \le\  \left\| Y_0 \right\|_\infty + c_0 M .
\ee

Note that equation \eqref{B} is all we need to solve problem \eqref{modelreduced} plugging its solution $Z(t)$ into formula \eqref{Y}. Moreover, in the cost function \eqref{cost} the variable $Y(x,t)$ appears only through $Z(t)$ (see \eqref{zeta}), thus it reads
\be\label{nuovacost}
\beq{l}
\dis \Phi(c,v) = A_0 \int_0^{T} c(t) s(t) Z(t) dt + A_1\int_0^{T}  Q\left(c(t)\right)  dt\\\\
\dis \hskip5cm+ \int_0^{T} \left(A_{21} v(t)+ \frac{A_{22}}{2} v^2(t) \right) dt .
\eeq
\ee
It is then convenient to adopt the following state system on the variables $s(t)$ and $Z(t)$
\be\label{statesys}
\left\{
\beq{rl}
i)&\dis s'(t) = -v(t) + \d (1 -s(t)) ,\\\\
ii)&\dis Z(t)= \int_{0}^{t} \mathcal{K}(t-x,x) Z(t-x)dx+F(t) .
\eeq
\right.
\ee
Thus, in the following we will focus on minimizing \eqref{nuovacost} under the state system \eqref{statesys}.

We now prove
\begin{theorem}
Let the assumptions \eqref{beta}-\eqref{Y0},\eqref{Q}, be satisfied, then there exists at least one optimal control $(c^* , v^* )$ minimizing the cost function \eqref{nuovacost} .
\end{theorem}
\begin{proof}
Since $\Phi (c,v)\geq 0$, it
follows that $\Phi$ has an infimum $d$ which is nonnegative. Let us take a
minimizing sequence $(c_{n},v_{n})\in \mathcal{U}$ such that\begin{equation}
d\leq \Phi (c_{n},v_{n})\leq d+\frac{1}{n},\mbox{ for }n\geq 1.
\label{d}
\end{equation}
and denote by $(s_{n}, Z_{n})$  the solution to the state system \eqref{statesys}  corresponding to $(c_{n},v_{n})$; namely 
\be\label{sn}
s_n(t)=e^{-\delta t}+\int_{0}^{t}e^{-\delta (t-\s)}(\delta -v_n(\s))d\s,\quad t\in [0 , T] .
\ee
and $Z_n(t)$ satifies
\be\label{Zn}
Z_n(t)= \int_{0}^{t} c_n(x) s_n(x) \beta _{0}(t-x) \Gamma (t-x) Z_n(x)dx+F(t) .
\ee

By \eqref{stime}-\eqref{stime1} it follows that on a subsequence we have
\be\label{convergenze}
\left\{
\beq{ll}
\dis c_n \to c^* & \mbox{ weak-* in }L^{\infty }(0,T)\\\\
\dis v_n \to v^*  &\mbox{ weak-* in }L^{\infty }(0,T)\\\\
\dis s_{n}\rightarrow s^*&\mbox{uniformly in }[0,T],\\\\
Z_{n}\rightarrow Z^{\ast }&\mbox{uniformly in }[0,T],
\eeq
\right.
\ee
with $(c^* , v^*) \in \cU$ because $\cU$ is  weak-* closed. Then we can go to the limit in \eqref{sn}-\eqref{Zn} and prove that  ($s^*, Z^*$), is indeed a solution to the state system.

Next, we pass to the limit in (\ref{d}), by calculating separately some terms in $\Phi (c_{n},v_{n})$. Indeed,  first the relations \eqref{convergenze} imply that 
\begin{equation}
\lim_{n\rightarrow \infty
}\int_{0}^{T}c_{n}(t)s_{n}(t)Z_{n}(t) dt=\int_{0}^{T}c^{\ast }(t)s^{\ast }(t) Z^{\ast }(t) dt=:I_{1}.  \label{I1n}
\end{equation}
Relying on the weakly lower semicontinuity property of the function $\dis v\rightarrow \int_{0}^{T}v^{2}(t)dt$ we have 
\[
\liminf_{n\rightarrow \infty }\int_{0}^{T}v_{n}^{2}(t)dt\geq
\int_{0}^{T}{v^{\ast }}^{2}(t)dt 
\]
and so 
\begin{equation} \label{I3n}
\beq{l}
\dis \liminf_{n\rightarrow \infty }\int_{0}^{T}\left( A_{21}v_{n}(t)+\frac{1}{2}
A_{22}v_{n}^{2}(t)\right) dt\\\\
\dis \hskip3.5cm\geq \int_{0}^{T}\left( A_{21}v^{\ast }(t)+\frac{
1}{2}A_{22} {v^{\ast }}^{2}(t)\right) dt=:I_{2}. 
\eeq
\end{equation}

Finally, for the remaining term of $\Phi(c_{n},v_{n})$ we have
\[
 \liminf_{n\rightarrow \infty
}\int_{0}^{T}Q(c_{n}(t))dt\geq \int_{0}^{T}Q(c^{\ast }(t))dt=:I_{3} ,
\]
where we used the lower semicontinuity of $Q$ (see \eqref{Q}).

Collecting the inequalities above, we conclude that, going back to (\ref{d}) we have 
\[
\Phi (c^{\ast },v^{\ast })=I_{1}+I_{2}+I_{3}\leq
\liminf_{n\rightarrow \infty }\Phi (c_{n},v_{n})\leq d, 
\]
that is $d\leq \Phi (c^{\ast },v^{\ast })\leq d$, which proves
that $(c^{\ast },v^{\ast })$ is optimal.
\end{proof}

\section{First-order necessary conditions of optimality}\label{opcond}

In this section we derive the necessary conditions a couple $(c^{\ast },v^{\ast })$ must satisfy in order to be optimal. To this end we strengthen condition \eqref{Q} assuming
\begin{equation}\label{assu-Q}
Q\in C^{1}([c_{-},c_{0}]),\quad Q\ \ \hbox{is strictly convex}, \ \ c_- > 0 .
\end{equation}

To identify the optimality conditions,  we now preliminarily determine the system in variation related to the state system \eqref{statesys},  considering an optimal couple $(c^* , v^*)$ and the variation $(c^\l,v^\l) \in \cU$, with
\be\label{variazioni}
\beq{ll}
c^{\lambda }(t):=c^{\ast }(t)+\lambda \varsigma(t) ,&\varsigma :=\widetilde{c}-c^{\ast },\\\\
v^{\lambda }(t):=v^{\ast }(t)+\lambda \omega(t) ,&\omega :=\widetilde{v}-v^{\ast },
\eeq
\ee
where $\l\in (0,1)$ and $( \widetilde c,\widetilde v)\in \cU$. Note that , since
\[
 \widetilde v(t)=v^\l(t)=0,\ \ \hbox{a.e. in}\ \ [0,T_1],
\] 
then 
\be\label{nuovo}
 \omega(t) =0,\ \ \hbox{a.e. in}\ \ [0,T_1] .
\ee

Corresponding to $(c^{*},v^{*})$ and $(c^{\lambda },v^{\lambda })$, the state system (\ref{statesys})
has unique solutions that we respectively denote $(s^{*},Z^{*})$ and $(s^{\lambda },Z^{\lambda })$. Then we consider the limits
\begin{equation} \label{state-var}
z=\lim_{\lambda \rightarrow 0}\frac{s^{\lambda }-s^{\ast }}{\lambda },\quad D=\lim_{\lambda \rightarrow 0}\frac{Z^{\lambda }-Z^{\ast }}{\lambda }.
\end{equation}
and, on the basis of (\ref{s}) and (\ref{B}), we  prove 
\begin{proposition}\label{pro1}
Let the assumptions \eqref{beta}-\eqref{Y0}  hold, then under the variation \eqref{variazioni}, the limits \eqref{state-var} exist in $C([0,T])$ and satisfy the system in variations
\begin{equation} \label{syst-var}
\left\{ 
\beq{rl}
i)&\dis z^{\prime }(t)=-\omega (t)-\delta z(t),\quad z(0)=0 , \\\\ 
ii)&\dis D(t) = \int_0^t  \mathcal{K}^*(x , t-x ) D(x) dx  +  \int_0^t f_{var}(x)\b_0(t-x) \Gamma(t-x) dx ,
\eeq
\right. 
\end{equation}
with 
\be\label{fvar}
f_{var}(t)=\left( c^{\ast }(t)z(t)+\varsigma (t)s^{\ast }(t)\right) Z^*(t). 
\ee
\end{proposition}
\begin{proof}
Indeed, \eqref{s} yields
\[
\frac{s^{\lambda }-s^{\ast }}{\lambda }=z(t)=-\int_{0}^{t}e^{-\delta (t-s)}\omega (s)ds, 
\]
equivalent to (\ref{syst-var}, $i$). Besides, from \eqref{B} we have
\begin{equation}\label{Blambda}
Z^\l(t)= \int_{0}^{t} \mathcal{K}^\l(t-x,x) Z^\l(t-x)dx+F(t),  
\end{equation}
\begin{equation} \label{Bstar}
Z^*(t)= \int_{0}^{t} \mathcal{K}^*(t-x,x) Z^*(t-x)dx+F(t), 
\end{equation}
with
\[
\mathcal{K}^\l(\s,x )=c^\l(\s) s^\l(\s) \beta _{0}(x) \Gamma (x) ,\quad \mathcal{K}^*(\s,x )=c^*(\s) s^*(\s) \beta _{0}(x) \Gamma (x) 
\]
(here again, $ \beta _{0}(x)$ and  $\Gamma (x)$ are extended by zero). Then the variable 
\[
\dis D^\l(t) =\frac{Z^\l (t)- Z^*(t)}\l
\]
satisfies
\[
D^\l(t) = \int_0^t  \mathcal{K}^\l(x , t-x ) D^\l(x) dx  + \int_0^t \frac { \mathcal{K}^\l(x, t-x )  -  \mathcal{K}^*(x , t-x) }\l Z^*(x) dx .
\]
Since
\[
\lim_{\l \to 0} \mathcal{K}^\l(\s,x )= \mathcal{K}^*(\s,x ),
\]
and
\[
\lim_{\l \to 0} \frac{\mathcal{K}^\l(\s,x )- \mathcal{K}^*(\s,x )} \l =  \left[ \varsigma(\s) s^*(\s) + c^*(\s) z(\s) \right] \b_0(x)\Gamma(x),
\]
in $L^\infty((0,T)\times (0,+\infty))$, we have (see \eqref{fvar})
\[
\lim_{\l \to 0}\int_0^t \frac { \mathcal{K}^\l(x, t-x )  -  \mathcal{K}^*(x , t-x) }\l Z^*(x) dx =  \int_0^t f_{var}(x)\b_0(t-x) \Gamma(t-x) dx, 
\]
in $C([0,T])$.  Then, by Theorem \ref{propVIE} ,we have that  $\dis \lim_{\l \to 0}  D^\l = D$ in $C([0,T])$, where $D$ satisfies (\ref{syst-var}, $ii$).
\end{proof}

Concerning the integral equation (\ref{syst-var}, $ii$), we need to remark that its solution $D(t)$ satisfies the following equality
\be\label{dualeq}
\beq{l}
\dis \int_0^T \left[  h(t) - \int_t^T \cK^*(t , x-t) h(x) dx \right] D(t) dt \\\\
\dis \hskip4cm= \int_0^T f_{var}(t)\int_t^T \b_0(\s - t) \Gamma(\s - t) h(\s) d\s dt
\eeq
\ee
for any function $h \in C([0,T])$. Indeed, we can check \eqref{dualeq} integrating (\ref{syst-var}, $ii$) multiplied by $h(t)$ over the interval $[0,T]$ and then exchanging the integrals. Actually, together with equation (\ref{syst-var}, $i$), this equality inspires the following dual system on the variables $(p(t),q(t))$ respectively corresponding to the state variables $(s(t),Z(t))$
\begin{equation} \label{dual}
\left\{ 
\beq{rll}
i)&\dis -p^{\prime }(t)+\delta p(t)= q(t) c^{\ast }(t)Z^*(t), & p(T)=0\\ \\ 
ii)&\dis q(t) = \int_t^T \cK^*(x, x-t) q(x) dx  +  A_0 .
\eeq
\right. 
\end{equation}

Concerning this problem we have:
\begin{proposition}
Let the assumptions \eqref{beta}-\eqref{Y0}  be satisfied, then for any couple $(c^*, v^*)\in \cU$, system \eqref{dual} has a unique solution $(p,q) \in C([0,T])\times C([0,T])$.
\end{proposition}
\begin{proof}
Noticing that equation (\ref{dual}, $ii$) is independent of $p(t)$, we can consider it separately and then plug $q(t)$ into 
(\ref{dual}, $i$).  Actually, (\ref{dual}, $ii$) can be conveniently transformed using the variable $u(t)=q(T-t)$, obtaining
\[
u(t)= \int_0^t \cK^*(T-x , t-x) u(x) dx  + A_0 .
\]
By Theorem \ref{VIE} a unique solution $u\in C([0,T])$ exists  and $q(t)= u(T-t)$ is the unique solution to (\ref{dual}, $ii$).

Finally we use $q(t)$ in (\ref{dual}, $i$) and solve explicitly for $p(t)$ obtaining
\be\label{p}
p(t) =\int_t^T e^{\delta(t-\s)} q(\s) c^{\ast }(\s)Z^*(\s) d\s ,
\ee
which is continuous.
\end{proof}

Before stating the main theorem of this section, we consider the multivalued function (actually the maximal monotone graph representing the normal cone of the interval $[a,b]$), $N_{[a,b]}: [a,b] \to 2^{\erre}$ defined as
\[
N_{[a,b]}(x) = \left\{ 
\beq{lll}
\dis (-\infty ,  0]&\hbox{for}& x=a\\
\dis 0&\hbox{for}& x\in (a , b)\\
\dis [ 0,\infty)&\hbox{for}& x=b ,\\
\eeq
\right.
\]
and recall the following Lemma (see \cite{Barbu1994})
\begin{lemma}\label{lemma}
Consider the convex set $C_{a,b}\subset L^2(T_1,T_2)$ defined as
\[
C_{a,b} \equiv \left\{ u \in L^2(T_1,T_2) ;\  a \le u(t) \le b,\ \hbox{ a.e. in}\  [T_1,T_2] \right\},
\]
then the normal cone of $C_{a,b}$ at $u\in C_{a,b}$ is given by
\[
\cN_{a,b} = \left\{ v \in L^2(T_1,T_2) ;\  v(t) \in N_{a,b}(u(t)), \ \hbox{ a.e. in}\  [T_1,T_2]\right\}.
\]
If moreover $f : [a,b] \to \erre$ is any strictly increasing function defined in $[a,b]$, then, defining
\[
\cF : C_{a,b} \to L^2(T_1,T_2),\quad \cF(u)(t) = f(u(t)),\ \ t\in [T_1,T_2]  
\]
we have $\left( \cF + \cN_{a,b}\right)^{-1} : L^2(T_1,T_2) \to C_{a,b}$ and
\[
\left( \cF + \cN_{a,b}\right)^{-1}(u)(t) = \cL[f,a,b](u(t)),\quad \hbox{a.e. in}\ \ [T_1,T_2],\ \ \forall u\in L^2(T_1,T_2),
\]
where the function $\cL[f,a,b] : \erre \to \erre$ is defined as
\[
\cL[f,a,b](x) \equiv \left\{
\beq{lll}
a&\hbox{for}& x < f(a)\\
f^{-1}(x)&\hbox{for}& f(a) \le x \le f(b)\\
b&\hbox{for}& x > f(b)
\eeq
\right. .
\]
\end{lemma}

Noticing that by \eqref{assu-Q}  $Q'$ is increasing in $[c_- , c_0]$ we consider the increasing functions 
\[
\beq{ll}
\dis f_1: [c_- , c_0] \to \erre,&f_1(x)= A_1 Q'(x) , \\\\
\dis f_2: [0, \delta] \to \erre,& f_2(x)= A_{22} x .
\eeq
\]
Then we have
\begin{theorem}\label{optimaltheorem}
Let the assumptions  \eqref{beta}-\eqref{Y0}, \eqref{assu-Q} hold and let $(c^{\ast },v^{\ast })$
be an optimal couple. Then, the first-order necessary conditions of optimality read, 
\be\label{cond}
\left\{
\beq{rl}
i)&  c^*(t) = \cL[f_1,c_-,c_0](E_1(t))\   \hbox{a.e. in}\ [0, T],\\\\
ii)&\dis v^*(t)=0,\ \hbox{a.e. in}\ [0,T_1],\ \ v^*(t) = \cL[f_2,0,\delta](E_2(t))\ \hbox{a.e. in}\ [T_1, T],
\eeq
\right.
\ee
where
\be \label{E} 
\beq{lll}
E_{1}(t) &=&\dis - q(t) s^{\ast }(t) Z^*(t), \\\\
E_{2}(t) &=& p(t)-A_{21},  
\eeq
\ee
and $(p,q)$ is the solution to the backward dual system \eqref{dual} corresponding to $(c^{\ast },v^{\ast })$.
\end{theorem}
\begin{proof}
We start performing the variations  \eqref{variazioni}, for which we have
\begin{equation} \label{cond-opt-lam}
\Phi (c^{\ast },v^{\ast })\leq \Phi (c^{\lambda},v^{\l}), 
\end{equation}
so that (see \eqref{state-var})
\be\label{cal-3}
\beq{l}
\dis \lim_{\lambda \rightarrow 0}\frac{\Phi (c^{\lambda
},v^{\l})-\Phi (c^{\ast },v^{\ast })}{\lambda } \\\\
\dis\hskip 1.5cm =\int_{0}^{T}A_{0}s^{\ast }(t) Z^*(t)  \varsigma (t) dt 
+\int_{0}^{T}A_{0}c^{\ast }(t)Z^*(t) z(t) dt\\\\
\dis\hskip2.5cm +  \int_{0}^{T}A_{0}c^{\ast }(t)s^{\ast }(t)D(t) dt  +A_{1}\int_{0}^{T}Q^{\prime }(c^{\ast }(t))\varsigma
(t)dt \\\\ \dis \hskip5.5cm+  \int_{0}^{T}(A_{21}+A_{22}v^{\ast }(t))\omega (t)dt\ \geq 0.
\eeq
\ee
In order to get rid of $z$ and $y$ we first multiply (\ref{syst-var}, $i$) by $p$ and integrate over $[0,T]$ obtaining
\be\label{cal-1}
\beq{ll}
- \dis\int_{0}^{T} \omega(t) p(t) dt &=\dis \int_{0}^{T}\left[ z^{\prime }(t)+\delta z(t)\right] p(t) dt\\\\
&=\dis\int_{0}^{T}\left[ -p^{\prime }(t)+\delta p(t) \right] z(t)dt\\\\
& = \dis \int_{0}^{T} q(t) c^{\ast}(t)Z^*(t) z(t)dt
\eeq
\ee
where we have integrated by parts using the initial conditions for $z$ and the final  condition for $p$ and then used equation (\ref{dual} , $i$). Beside we consider equation (\ref{syst-var}, $ii$)  and the consequent equality \eqref{dualeq}, choosing  $h(t)= c^*(t) s^*(t) q(t)$, where  $q$ is the solution of  (\ref{dual} , $ii$). We obtain 
\[
\beq{l}
\dis \int_0^T c^*(t) s^*(t) \left[  q(t) - \int_t^T \cK^*(\s, \s-t) q(\s) d\s \right] D(t) dt \\\\
\dis\hskip2cm = \int_0^T  f_{var} (t)  \int_t^T \b_0(\s - t)\Gamma(\s-t)c^*(\s) s^*(\s) q(\s) d\s dt\\\\
\dis\hskip2cm = \int_0^T  f_{var} (t)  \int_t^T \cK^*(\s , \s -t)q(\s) d\s dt .
\eeq
\]
then, using (\ref{dual} , $ii$) we conclude with (see \eqref{fvar})
\be\label{cal-2}
\int_{0}^{T}A_{0}c^{\ast }(t)s^{\ast }(t)D(t) dt  =  \int_0^T \left( c^{\ast }(t)z(t)+\varsigma (t)s^{\ast }(t)\right) Z^*(t) (q(t) - A_0) dt .
\ee
Putting together \eqref{cal-1} and \eqref{cal-2} we finally obtain
\[
\beq{l}
\dis \int_{0}^{T}A_{0}c^{\ast }(t)Z^*(t)z(t) dt+ \int_{0}^{T}A_{0}  c^{\ast }(t)s^{\ast }(t)D(t) dt  \\\\
\dis \hskip1cm=-\int_{0}^{T}\omega (t)p(t)dt+\int_{0}^{T}\varsigma(t)s^{\ast }(t)Z^*(t) (q(t)-A_0) dt. 
\eeq
\]
Thus, substituting into \eqref{cal-3} we have
\[
\beq{l}
\dis- \int_{0}^{T}\omega (t)p(t)+\int_0^T \varsigma(t)s^{\ast }(t)Z^*(t) q(t) dt \\
\dis \hskip2cm+A_{1}\int_{0}^{T}Q^{\prime }(c^{\ast }(t))\varsigma
(t)dt+\int_{0}^{T}(A_{21}+A_{22}v^{\ast }(t))\omega (t)dt\geq 0
\eeq
\]
and, using \eqref{variazioni} and \eqref{nuovo} we conclude that
\be\label{cal-4}
\beq{l}
\dis\int_{0}^{T} \left[ - s^{\ast}(t)Z^*(t) q(t) -A_{1}Q^{\prime }(c^{\ast }(t))\right] (c^*(t)-\widetilde c (t)) dt  \\\\
\dis \hskip3cm+\int_{T_1}^{T}\left[ p(t)-A_{21}-A_{22}v^{\ast }(t)\right] (v^*(t)-\widetilde v (t))  dt \ \geq \ 0,
\eeq
\ee
for all $\widetilde c$, $\widetilde v$ such that
\[
c_- \le \widetilde c(t) \le c_0, \ \ \hbox{a.e. in}\ \ [0, T],\quad 0 \le \widetilde v(t) \le \delta, \ \ \hbox{a.e. in}\ \ [T_1, T] .
\]
In particular, for $\widetilde{v}=v^{\ast }$, we deduce that 
\begin{equation} \label{cal-5-0}
 E_1(\cdot)-A_{1}Q^{\prime }(c^{\ast }(\cdot))\in \cN_{[c_-,c_{0}]}(c^{\ast})
\end{equation}
where $\cN_{[c_{-},c_{0}]}(c^{\ast })\subset L^2(0,T)$ is the normal
cone to the convex set
\[
C_{[c_{-},c_{0}]}\equiv \left\{ u \in L^2(0,T) ;\  c_- \le u(t) \le c_0,\ \hbox{ a.e. in}\ \  [0,T] \right\}
\] 
at $c^{\ast }$. This can be still written 
\begin{equation} \label{cal-6}
E_1(t) \in f_1(c^*(t))+
N_{[c_{-},c_{0}]}(c^{\ast}(t)),\quad \hbox{a.e.  in}\ \ [0 , T]
\end{equation}
and (\ref{cond}, $i$) follows from Lemma \ref{lemma}.

Analogously, going back to (\ref{cal-4}) and setting $\widetilde{c}=c^{\ast }$ we obtain 
\begin{equation} \label{cal-11}
p(\cdot)-A_{21}-A_{22}v^{\ast }(\cdot)\in \cN_{[0,1]}(v^{\ast }), 
\end{equation}
or still
\[
E_{2}(t)\in   f_2(v^*(t))+ N_{[0,\delta]}(v^{\ast }(t)), \quad \hbox{a.e.  in}\ \ [T_1 , T]
\]
and (\ref{cond}, $ii$) is proved.
\end{proof}

\section{A worked example to discuss optimality conditions}
The optimality conditions \eqref{cond}, stated in Theorem \ref{optimaltheorem}, provide  a functional system of equations to determine the shape  of the optimal strategy $(c^*,v^*)$. Indeed, these equations can be used in connection with a numerical method to determine the possible strategies under realistic values of the parameters. However, some informations can be obtained analysing the behavior of the functions $E_1(t)$ and $E_2(t)$,  as far as possible independently of the specific form of the seeked  functions $c^*(t)$ and $v^*(t)$. A specific worked example may help to show how this analysis can be carried through, at least indicating some possible scenarios to be fully clarified by numerical methods.

To this aim we consider the case $x^+ = +\infty$ with the following parametrization
\[
\b_0(x) = \bar\b e^{-\varphi x},\quad \g(x) = \g
\]
where $\varphi$ and $\gamma$ are positive constants. The former assumption implies that infectiousness is maximal during the earlier stage of infectivity development and then declines exponentially. Moreover, as initial density of infected we take the corresponding stable age distribution of free epidemic (see \eqref{stabledistribution}) 
\[
Y_0(x) = I_0 (\a^*_0 + \g) e^{-(\alpha^*_0 + \gamma)x}
\]
With these assumptions, setting $\theta= \varphi + \gamma$,  the state equation \eqref{B} reads
\[
Z(t) = \bar\b  e^{-\theta t} \left(\int_0^t c(x)s(x) e^{\theta x} Z(x) dx +  I_0 (\a^*_0 + \g) \int_0^\infty e^{-(\theta+\a_0^*)x}\right) ,
\]
which is equivalent to the  non-autonomous ordinary equation,
\[
Z'(t)=\bar\b c(t)s(t) Z(t) -\theta Z(t),\quad Z(0) = \bar\b I_0 \frac{\g +\alpha^*_0}{\theta+\alpha^*_0} ,
\]
The previous equation is the infective equation of a unstructured (SIR or SIRS) model where infectiousness decline supplies an additional source of removal, besides the baseline removal rate $\gamma$, generating the overall removal rate $\theta = \varphi + \gamma$. 

Based on previous assumptions
\be\label{Z}
Z^*(t) = Z(0) e^{-\theta t} e^{\bar\b \int_0^t c^*(\s)s^*(\s) d\s}.
\ee
Concerning the dual variable $q(t)$, a similar computation yields the (backward) ODE problem
\[
q'(t) = \left( \theta - \bar \b c^*(t) s^*(t)\right) q(t) - \theta A_0 ,\quad q(T)= A_0 ,
\]
so that
\be\label{q}
q(t) = A_0 e^{\theta t}\left( e^{-\theta T} e^{\bar\b \int_t^T c^*(r)s^*(r) dr} + \theta\int_t^Te^{-\theta \s}e^{-\bar\b \int_\s^tc^*(r)s^*(r) dr} d\s \right) .
\ee
Actually, we are interested to the product $\Pi(t)= q(t) Z^*(t)$ that reads
\[
\Pi(t)= A_0Z(0) \Pi_0(t)
\]
where
\be\label{Pi}
\Pi_0(t)=  \theta \int_t^T e^{-\theta \s} e^{\bar\b \int_0^\s c^*(r) s^*(r) dr} d\s+  e^{-\theta T} e^{\bar\b \int_0^T c^*(r)s^*(r) dr} ,
\ee
is a positive, decreasing function. Finally we have
\[
\left\{
\beq{l}
\dis E_1(t)= - \Pi(t) s^*(t) , \\\\
\dis E_2(t)= A_0 Z(0)\ e^{\d t} \int_t^T e^{-\d\s}\Pi_0(\s) c^*(\s) d\s -A_{21} .
\eeq
\right.
\]
Note that $E_1(t)$  is negative and increasing as far as $s^*(t)$ is non increasing. Indeed, the behavior of  $E_1(t)$ is driven by $s^*(t)$ which is definitely non increasing for $t\in[0,T^*]$, with $T^* > T_1$. In fact,  $s^*(t)=1$ for $t\in[0,T_1]$ because $v^*(t)=0$ (see \eqref{s}), then it is non increasing in a right neighborhood of $T_1$, but may be oscillating for $t\in[T^*, T]$, according with the variations of $v^*(t)$. Moreover from \eqref{Pi} we can observe that
\[
\beq{ll}
\dis E_1(0)= - A_0 Z(0) \Pi_0(0)&< - A_0 Z(0) e^{-\theta T} e^{\bar\b \int_0^T c^*(r)s^*(r) dr} \\\\
\dis &\le - A_0 Z(0) e^{-\theta T} e^{\bar\b \int_0^T c^*(r)s^*(r) dr} s^*(T) = E_1(T) .
\eeq
\]
Also the behavior of $E_2(t)$ is largely unknown since the first term is the product of an increasing term and a decreasing one. Nonetheless, it holds that
\[
E_2(0) > - A_{21} = E_2(T), \quad E_2'(T) = - \Pi_0(T) c^*(T) < 0 .
\]
\begin{figure}[h]
\includegraphics[width=6.5cm]{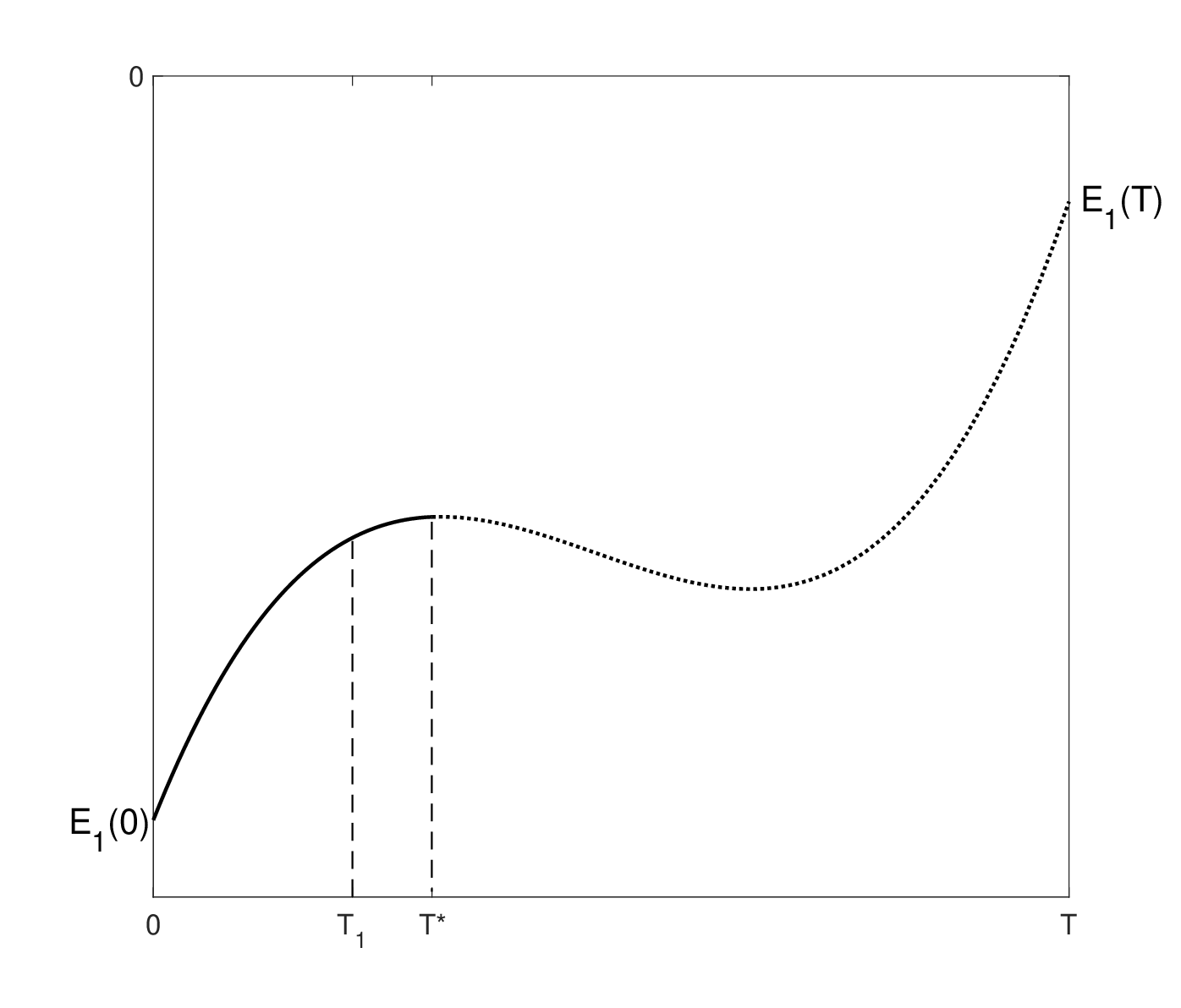}\includegraphics[width=6.5cm]{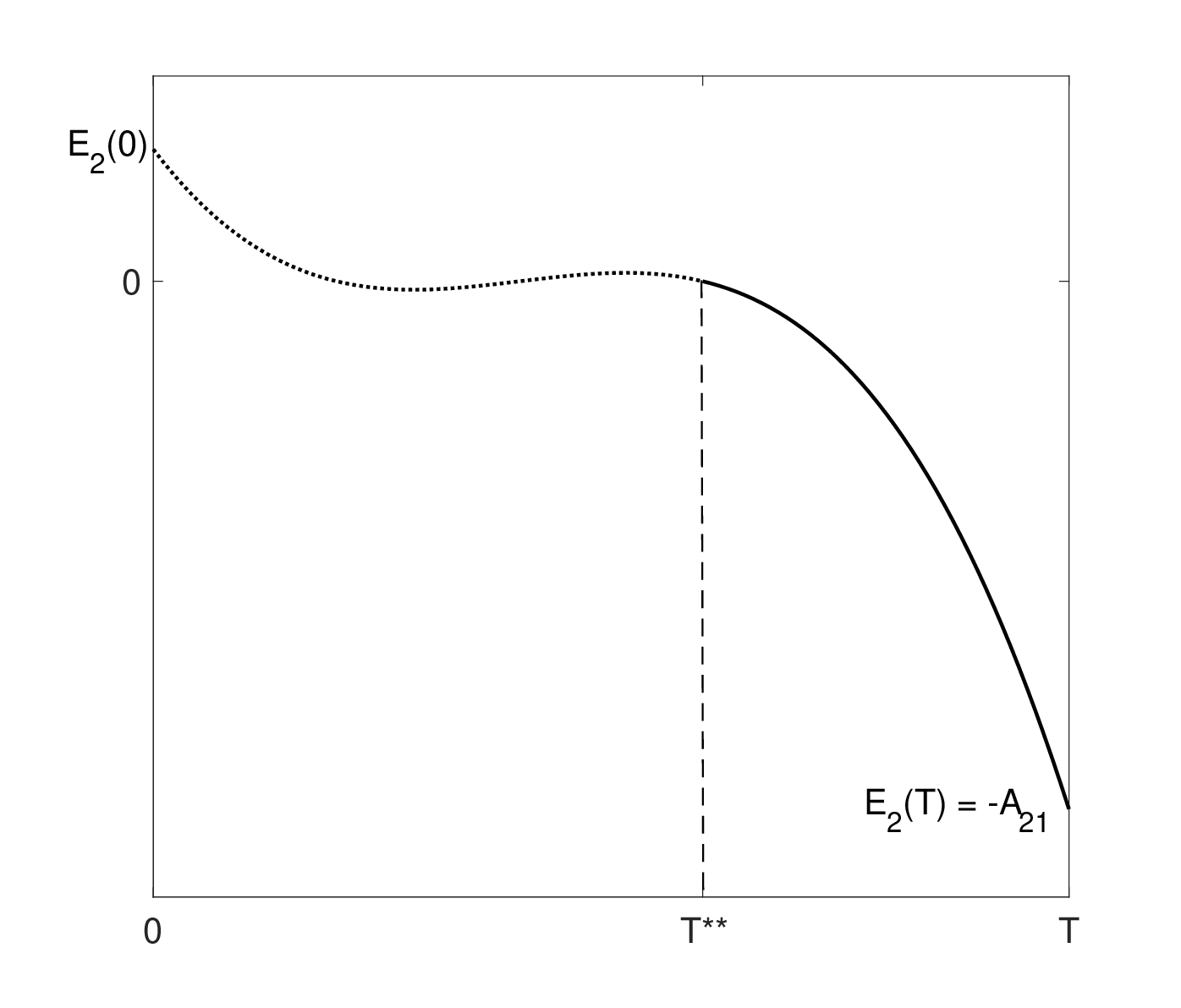}
\caption{\label{figura2}Possible behavior of $E_1(t)$ (left panel) and $E_2(t)$ (right panel).The solid line corresponds to analytically determined behavior, the dotted lines show only one possible behavior.}
\end{figure}
By collecting together the various information drawn so far, we obtain the graphs reported in Figure \ref{figura2}, where the dotted part of the curves show only one of the possible behaviors of functions $E_1(t)$ and $E_2(t)$. Nevertheless, in spite of the basically unknown behavior of functions $E_1(t)$and $E_2(t)$ we can draw some useful information on $c^*(t)$ and $v^*(t)$, based on  \eqref{cond}. In fact, function $c^*(t)$,  (\ref{cond}, $i$) may be restated as follows 
\be\label{opt1}
c^*(t)=\left\{
\beq{ll}
c_- &\hbox{if}\ \ E_1(t) <  A_1 Q'(c_-) ,\\\\
\dis (Q')^{-1} \left(\frac 1{A_{1}}E_1(t)\right)&\hbox{if}\ \ A_1 Q'(c_-) \le E_1(t) \le  A_1 Q'(c_0) ,\\\\
c_0&\hbox{if}\ \ A_1 Q'(c_0) < E_1(t) ,
\eeq
\right.
\ee
showing that $E_1(t)$ must be compared with the values
\[
h_1= A_1Q'(c_-), \quad h_2= A_1Q'(c_0), \qquad ( h_1 < h_2 \le 0) .
\]
This shows we can obtain informations about $c^*(t)$ in the initial and in the final portion of the interval $[0,T]$. Namely, we have
\be\label{bohc0}
\left\{
\beq{rl}
i)&\dis \hbox{If}\ \ E_1(0) <  h_1, \ \ \hbox{then} \quad c^*(t) = c_- \ \ \hbox{for}\ \ t\in [0, T^*]\\\\
ii)&\dis \hbox{If}\ \ E_1(0) >  h_2, \ \ \hbox{then} \quad c^*(t) = c_0 \ \ \hbox{for}\ \ t\in [0, T^*]
\eeq
\right.
\ee
and, similarly
\be\label{bohcT}
\left\{
\beq{rl}
i)&\dis \hbox{If}\ \ E_1(T) <  h_1, \ \ \hbox{then} \quad c^*(t) = c_- \ \ \hbox{for}\ \ t\in [T^{**}, T]\\\\
ii)&\dis \hbox{If}\ \ E_1(T) >  h_2, \ \ \hbox{then} \quad c^*(t) = c_0 \ \ \hbox{for}\ \ t\in [T^{**}, T]
\eeq
\right.
\ee
Note that, since
\[
E_1(0) = -  A_0 Z(0) \Pi_0(0)  < -  A_0 \bar\b I_0 \frac{\g +\alpha^*_0}{\theta+\alpha^*_0},
\]
then condition in (\ref{bohc0}, $i$) holds  if
\[
A_0 \bar\b I_0 \frac{\g +\alpha^*_0}{\theta+\alpha^*_0} >  -A_1Q'(c_-).
\]
The previous condition is satisfied if, for instance, $A_0$ is sufficiently larger than $A_1$, 
 i.e. when the weights attributed to the direct health cost of the epidemics is larger than the corresponding indirect loss. Indeed, in this case the optimal strategy corresponds to set the contact rate $c(t)$ to its minimum value $c_-$ in the early phase of epidemic intervention.  Conversely, similar estimates show that if $A_1 \gg A_0$  then (\ref{bohc0}, $ii$) holds. Similar considerations hold for  \eqref{bohcT}.

Finally, concerning function $v^*(t)$, we can draw some useful information by reformulating (\ref{cond}, $ii$)
as
\be\label{opt2}
v^*(t)=\left\{
\beq{ll}
0 &\hbox{if}\ \ E_2(t) <  0 ,\\\\
\dis\frac1{A_{22}}  E_2(t)&\hbox{if}\ \ 0 \le E_2(t) \le  A_{22} \delta ,\\\\
\delta &\hbox{if}\ \ A_{22} \delta< E_2(t) .
\eeq
\right. 
\ee
Then, for $v^*(t)$ in the interval $[T_1, T]$ we have
\be\label{bohv0}
\left\{
\beq{rl}
i)&\dis \hbox{If}\ \ E_2(0) <  0, \ \ \hbox{then} \quad v^*(t) = 0 \ \ \hbox{for}\ \ t\in [T_1, T^*]\\\\
ii)&\dis \hbox{If}\ \ E_2(0) >  A_{22} \delta, \ \ \hbox{then} \quad v^*(t) = \d \ \ \hbox{for}\ \ t\in [T_1, T^*]
\eeq
\right.
\ee
and
\be\label{bohvT}
\left\{
\beq{rl}
i)&\dis \hbox{If}\ \ E_2(T) <  0, \ \ \hbox{then} \quad v^*(t) = 0 \ \ \hbox{for}\ \ t\in [T^{**}, T]\\\\
ii)&\dis \hbox{If}\ \ E_2(T) >  A_{22} \delta, \ \ \hbox{then} \quad v^*(t) = \d \ \ \hbox{for}\ \ t\in [T^{**}, T] .
\eeq
\right.
\ee
Note that since  
\[
E_2(0) \le   A_0 Z(0) c_0 e^{\bar\b c_0 T} T - A_{21}
\]
then condition (\ref{bohv0}, $i$) is  satisfied if
\[
A_0 Z(0) c_0 e^{\bar\b c_0 T} T <   A_{21} .
\]
The latter condition is interesting because it pinpoints the critical role played, in this case, by the relative (baseline) cost of the vaccination campaign compared to the health cost of the epidemic, given by the ratio $A_{21} /  A_{0}$.  Analogously, specific estimates show how conditions (\ref{bohv0}, $ii$), (\ref{bohvT}, $i$), (\ref{bohvT}, $ii$) can be satisfied by different combinations of the parameters.
 
\section{Some remarks}\label{remarks}
In this section we collect a number of remarks on the proposed results that are useful for both the interpretation of their meaning as well as the discussion of the many special subcases that can extracted from our general framework.

\subsection{Epidemiological framework: special subcases}
The epidemiological models considered in this work assumed a fairly general structure (SIRS). From this standpoint, our results can be easily specialised to include the noteworthy subcase of infections imparting full immunity after infection or vaccination by a perfect vaccine that is, the SIR case) which arises for $\delta=0$.\\
Additionally, further noteworthy modelling subcases arise by making suitable hypotheses on epidemiological parameters in order to allow the reduction of our general problem to more treatable ones. For example, by suitable hypotheses on the infectivity kernel and recovery functions $\beta_0(x), \gamma(x)$, the state equation of prevalence can be reduced to either fixed-delay differential equations (DDE) or even ordinary differential equations (ODE), allowing a greater deal of tractability \cite{macdonald,d2021dynamics}.

\subsection{Vaccination schedule: special subcases}
The approach proposed here is highly general and therefore allows to tackle many possible subcases. In relation to COVID-19, and in general to pandemic control, or other urgency situations due to a new virus, a critically important case is that the vaccination control, besides becoming available in a later stage of the epidemic, was subject to a number of constraints that prevent to use it as an optimizable resource. For COVID-19 this has been rule even in industrialised countries mainly due to \emph{supply side constraints}, such as limited availability, bottlenecks in the distribution and administration chains, etc but also to constraints on the \emph{demand side }, such as - besides vaccine hesitancy - the need to giving full priority to population groups at high risk of serious sequelae, or to specific categories, such as medical doctors and nurses, etc). For example, in the case of COVID-19, the best that could be obtained was to achieve, in a span of months, a maximal number of daily vaccine administrations, given constraints on logistics (spaces) and vaccinating staff.\\
For the aforementioned situations, function $V(t)$ is to be taken as given over the time horizon considered, so that the optimization issue re-collapses into a one-control problem, where social distancing is the only optimizable resource and the arrival of the vaccine simply allows to re-optimise given the (known or expected) temporal profile $\overline{V}(t)$ of the schedule of vaccination administration.
For this case our general results may be used to consider the case in which $v(t)$ is assigned and the set of controls includes only possible strategies for $c(t)$. Indeed, in this case only condition \eqref{opt1} is active and allows to determine the optimal $c^*(t)$.\\
As for the modelling of the vaccination schedule, an alternative parametrization avoiding the violation of the positivity of $S(t)$, is $V(t)=\sigma(S(t)) S(t)$ where $\sigma(S(t))$ is a time-dependent, possibly nonlinear, per-capita rate of successful immunization. The latter formulation has the advantage of being amenable to a variety of behavioral formulations \cite{manfredi2013modeling}.

\subsection{Realistic vaccine waning}
Equation (\ref{modelreduced}, $i$), including immunization by a vaccine with waning immunity, assumes the simplest form of waning immunity i.e., exponential waning a constant waning rate $\delta$. This is acceptable as a departure point in the absence of detailed informations on the building of vaccine-related immunity and the subsequent waning process. A more structured description of such a process is easily embedded in the present framework by resorting again to a structured equation using an additional variable  $\a$  to represent the time  elapsed since the instant of vaccine administration. This implies to consider a density function of vaccinated individuals as 
\[
\cI(\a, t),\quad \a \in [0 , \a^+],\ t\in [0,T]
\] 
In this way,  equation  (\ref{modelreduced}, $i$) would be replaced by
\[
\left\{
\beq{rl}
i)&\dis s'(t) = -v(t) + \int_0^{\a^+} \d(\a) \cI(\a , t) d \a\\\\
ii)&\dis \left(\dpar t{} + \dpar x{}\right) \cI(\a,t) = -\delta(\a)\cI(\a,t),\\\\
iii)&\dis \cI(0,t)=  v(t) .
\eeq
\right.
\]

This realistic representation adds a layer of complexity to our but still remaining within the domain of structured epidemic equations, which could still allow a unified treatment. 

\section{Concluding remarks}
This paper has considered the problem of optimally controlling an epidemic of an infectious disease not yielding permament immunity (i.e., of an SIRS type) and structured by time since infection. The main novelty of this work is represented by the consideration of two control tools, namely the enaction of social distancing and vaccination, respectively. The problem was tackled by a fully general theoretical approach by which we could (i) prove the existence of (at least) one optimal control pair, (ii) derived the first-order necessary conditions for optimality, (iii) demonstrate some useful properties of the optimal solutions. The generality of the proposed approach and findings allows, in principle, to treat a number of relevant subcases that have emerged during the COVID-19 epidemic such as e.g., (i) social distancing as the only available control tool during a first epoch of the pandemic, (ii) the arrival of a vaccine only in a second stage of the epidemic, (iii) problems of rationing in the vaccine supply, making social distancing the only optimizable instrument during the vaccination epoch. Note that the implementation of the theory developed here to actual problems, will require - in view of the peculiarity of the problem considered - the development of \emph{ad-hoc}
numerical methods, which will be a first task of our subsequent work.


\end{document}